\documentclass[11pt,oneside]{article}

\usepackage[round,colon,authoryear]{natbib} 

\usepackage{graphicx}
\usepackage{amsfonts}
\usepackage{amssymb}
\usepackage{amsmath}
\usepackage{amsthm}
\usepackage{enumerate}
\usepackage{times}
\usepackage{eurosym}
\usepackage{url}
\usepackage{bm}
\usepackage{comment}

\newtheorem{theorem}{Theorem}

\newtheorem{lemma}{Lemma}

\newtheorem{proposition}{Proposition}
\theoremstyle{remark}
\newtheorem{remark}{Remark}
\theoremstyle{remark}
\newtheorem{ex}{Example}
\theoremstyle{remark}
\newtheorem{definition}{Definition}
\theoremstyle{theorem}

\newcommand{\E}{{\mathbb{E}}}

\newcommand{\Prob}{{\mathbb{P}}}
\newcommand{\Q}{{\mathbb{Q}}}

\newcommand{\1}{\mathbf{1}}
\newcommand{\F}{\mathcal{F}}

\newcommand{\dd}{\mathrm{d}}
\begin{document}

\title{A Systematic Approach to Constructing Market Models With Arbitrage\thanks{%
This project started at the Sino-French Workshop at the Beijing International Center for Mathematical Research in June 2013.
We are deeply indebted to Ying Jiao, Caroline Hillairet, and Peter Tankov for organising this wonderful meeting and we are very grateful to all participants of the workshop for stimulating discussions on the subject matter of this note. We thank Claudio Fontana and Chau Ngoc Huy for many helpful comments on an earlier version of this note.}}
\author{Johannes Ruf\thanks{%
E-Mail: johannes.ruf@oxford-man.ox.ac.uk} \\
Oxford-Man Institute of Quantitative Finance and Mathematical Institute\\
University of Oxford, United Kingdom \and Wolfgang Runggaldier\thanks{%
E-Mail: runggal@math.unipd.it} \\Department of Mathematics\\University of Padova, Italy}
\maketitle

\begin{abstract}
	This short note provides a systematic construction of market models without unbounded profits but with arbitrage opportunities.
\end{abstract}



\section{Introduction}
One of the fundamental notions in modern mathematical finance is that of absence of arbitrage.  In fact, without no-arbitrage conditions one cannot meaningfully solve problems of pricing, hedging or portfolio optimization. A fundamental step in the historical development of the no-arbitrage theory was made by Harrison, Kreps, and Pliska in a series of papers, see \citet{HarrisonKreps}, \citet{Kreps_1981}, and \citet{HarrisonPliska}.
\cite{DS_fundamental, DS_1998} (see also \citet{Kabanov}) then proved the equivalence between the economic notion of No Free Lunch With Vanishing Risk (NFLVR)  and the mathematical concept of an Equivalent Local Martingale Measure (ELMM) in full generality. 

The more recent Stochastic Portfolio Theory (see, for example, \citet{Fe} or the survey in \cite{FK_survey}), which is a more descriptive rather than normative theory, shows that the behavior in real markets corresponds to weaker notions of no-arbitrage than NFLVR.  Somehow in parallel to this theory, the so-called Benchmark Approach to quantitative finance (see \citet{Pl2006} or the textbook \cite{PH}) was introduced with the aim of showing that pricing and hedging can also be performed without relying on the existence of an ELMM.

Various weaker notions of no-arbitrage have henceforth been studied as well as their consequences on asset pricing and portfolio optimization; for a recent unifying analysis of the whole spectrum of no-arbitrage conditions for continuous financial market models, see \cite{Fontana_2013}. A crucial concept in this development is the notion of an Equivalent Local Martingale Deflator (ELMD), which is the counterpart of the density process for the case when an ELMM exists. Like the density process, the ELMD is a local martingale, but it may fail to be a martingale.  In parallel with this theory, also a theory of asset price bubbles was developed where, under an ELMM, discounted asset prices are strict local martingales but we do not touch this issue here, and refer instead to \citet{JP_complete, JP_incomplete}.

Along the development of the weaker notions of no-arbitrage, a crucial step was made in the paper \cite{KK} where the authors show that a condition, which they call No Unbounded Profit With Bounded Risk (NUPBR), is the minimal condition for which portfolio optimization can be meaningfully performed.  For the corresponding hedging problem, we refer to \cite{Ruf_ots, Ruf_hedging}. The notion of NUPBR has also
 appeared under the name of No Arbitrage of the First Kind (NA1); see \citet{I} and \cite{Kardaras_2012_viability}. 
A related problem of interest, but that we do not deal with here, is that of the robustness of the concept of arbitrage, an attempt in this direction is made in \cite{Guasoni_Rasony_2011}.

As a consequence of the above, the interest arose in finding market models that fall between NFLVR and NUPBR. Such models would then allow for classical arbitrage, but make it still possible to perform pricing and hedging as well as portfolio optimization. A classical example for continuous market models appears already in \cite{DS_Bessel}. The relevance of Bessel processes in this context is also stressed in \cite{PH}; see also the note \citet{Ruf_2010_Bessel} and the survey \cite{Hulley_2010}.  At this point, one might then wonder whether there are other financially significant models, beyond those based on Bessel processes, that satisfy NUPBR but not NFLVR and whether there could be a systematic procedure to generate such models. Equivalently, whether there is a procedure to generate strict local martingales and the present paper is an attempt in this direction. 

Our approach is inspired by a recent renewed interest (see, for example, \citet{FK}, \citet{Ruf_Novikov}, \citet{CFR2011}, \citet{KKN_local}, \citet{Perkowski_Ruf_2013}) in the so-called F\"ollmer exit measure  of a strictly positive local martingale, constructed in \citet{F1972}, as it was already initiated for continuous processes in \cite{DS_Bessel} (see Theorem~1 there).   Another, but related point of view to look at our approach is to interpret the expectation process (as a function of time) of a nonnegative local martingale as the distribution function of a certain random variable, namely the time of explosion of a process that is related to the nonnegative martingale; this point of view is inspired by \citet{McKean_1969} and further explored in \citet{Karatzas_Ruf_2013}.

Our approach is closely related to the method suggested in \citet{OR}, where diverse markets are constructed through an absolutely continuous but not equivalent change of measure.
Parallel to this work, \citet{Protter_2013_strict} has developed an approach via a shrinkage of the underlying filtration to systematically obtain strict local martingales.  Using such an insight to construct strict local martingales might yield a further method to obtain models that satisfy NUPBR, but not NFLVR. We shall not pursue this direction and leave it open for future research.  On the contrary, \citet{Fontana_Jeanblanc_2013} construct models that satisfy NUPBR (at least, up to a certain time), but not NFLVR, via a filtration enlargement.

\section{The model and preliminary notions}

Given a finite time horizon $T<\infty$, consider a market with $d$ assets, namely a pair $(\Omega, \F, (\F(t))_{t \in [0,T]}, \Prob), S$ of a filtered probability space and a $d$--dimensional vector $S$ of nonnegative semimartingales  $(S_i)_{i = 1, \ldots, d}$ with $S_i = (S_i(t))_{t \geq 0}$.  We assume that $\F(0)$ is trivial and that the filtration $(\F(t))_{t \in [0,T]}$ is right-continuous. Each component of the process $S$ represents the price of one of $d$ assets that we assume already discounted with respect to the money market account; that is, we assume the short rate of interest to be zero.  Agents invest in this market according to a self-financing strategy $H=(H(t))_{t \in [0,T]}$ and we shall denote by $$V^{x,H}=(V^{x,H}(t))_{t \in [0,T]} = x+((H\cdot S)_t)_{t \in [0,T]}  =x+\left(\int_0^tH(u) \dd S(u)\right)_{t \in [0,T]}$$  the value process corresponding to the strategy $H$ with initial value $V^{x,H}(0)=x$.  

\begin{definition}\label{D1}  Let $\alpha>0$ be a positive number. An $S-$integrable predictable process $H$ is called {\sl $\alpha$--admissible} if $H(0)=0$ and the process $V^{0,H}$ satisfies $V^{0,H}(t) \geq -\alpha$ for all $t\in[0,T]$ almost surely.
The strategy $H$ is called {\sl admissible} if it is $\alpha$--admissible for some $\alpha>0$. \qed
\end{definition}

\begin{definition}\label{D2} An {\sl arbitrage strategy} $H$ is an admissible strategy for which $\Prob(V^{0,H}(T)\geq 0) =1$ and $\Prob(V^{0,H}(T)> 0) >0$. We call it a {\sl strong arbitrage strategy} if $\Prob(V^{0,H}(T)>0)=1.$ \qed
\end{definition}

We also recall that the classical notion of absence of arbitrage, namely \emph{No Free Lunch with Vanishing Risk (NFLVR)}, is equivalent to the existence of a probability measure $\Q$, equivalent to $\Prob$, under which the price processes are local martingales (as we assume the prices to be nonnegative).
Among the more recent weaker notions of absence of arbitrage we recall the following:

\begin{definition}\label{D3}  An $\F(T)$--measurable random variable $\xi$ is called an \emph{Arbitrage of the First Kind}  if $\Prob(\xi\ge 0)=1$, $\Prob(\xi> 0)>0$, and for all $x>0$ there exists an $x$--admissible strategy $H$ such that $V^{x,H}(T)\ge \xi$. We shall say that the market admits \emph{No Arbitrage of the First Kind (NA1)}, if there is no arbitrage of the first kind in the market. \qed\end{definition}

\begin{definition}\label{D3a} There is \emph{No Unbounded Profit With Bounded Risk (NUPBR)} if the set 
$$\mathcal{K}_1 = \left\{ V^{0,H}(T) \mid H = (H(t))_{t \in [0,T]} \text{ is a 1--admissible strategy for } S\right\}$$
is bounded in $L^0$, that is, if
$$\lim_{c \uparrow \infty} \sup_{W \in \mathcal{K}_1} \Prob(W>c) = 0$$
holds. \qed
\end{definition}

It can be shown that NA1 and NUPBR are equivalent (see \citet{Kardaras_finitely})  and that NFLVR implies NUPBR, but there is no equivalence between the latter two notions (see \citet{DS_fundamental} or \citet{KK}).
A market satisfying NA1 (or, equivalently, NUPBR) is also called (weakly) {\sl viable} and it can be shown that market viability in the sense of NA1 (NUPBR) is a minimal condition to meaningfully solve problems of pricing, hedging and portfolio optimization; see \citet{KK}.

The last notion to be recalled is that of an \emph{Equivalent Local Martingale Deflator (ELMD)}, which generalizes the notion of the density process for an ELMM:

\begin{definition}\label{D4} An \emph{Equivalent Local Martingale Deflator (ELMD)} is a nonnegative local martingale $Z$, not necessarily a martingale, such that $Z(0)=1$ and $\Prob(Z(T)>0)=1$, and the price processes, when multiplied by $Z$, become local martingales.
\qed
\end{definition}

The following result has only recently been proven in full generality; see  \citet{Kardaras_2012_viability}, \citet{Takaoka}, and \citet{Song_2013}:
\begin{proposition}  \label{P1}
 A market satisfies NUPBR if and only if the set of equivalent local martingale deflators is not empty. \end{proposition}

The goal of this note is now to provide a systematic way to construct a market that satisfies NUPBR, but not NFLVR.  

\section{Main result} \label{S:construction}

In this section we formulate two assumptions under which we can construct a market that satisfies NA1 (equivalently NUPBR) but not NFLVR. In the following Section~\ref{Examples}, we then present some examples for markets in which those assumptions are satisfied, namely for which NUPBR thus holds, but not NFLVR.

Based on \cite{DS_Bessel}, as it is said there, we now turn things upside down. While in the previous section we had started from a probability space $(\Omega, \F, (\F(t))_{t \in [0,T]}, \Prob)$, on which the $d$ asset price processes $S_i$ are semimartingales, now we consider a filtered probability space  $(\Omega, \F, (\F(t))_{t \in [0,T]}, \Q)$, on which the $d$ processes $S_i$ are $\Q-$local martingales. On this same probability space we then consider a further nonnegative $\Q$--martingale $Y = (Y(t))_{t \in[0,T]}$  with $Y(0)=1$. Let the stopping time $\tau$ denote the first hitting time of $0$ by the $\Q$--martingale $Y$. 
We shall assume that $Y$ has positive probability to hit zero, but it only hits zero continuously; to wit,
\begin{align}  \label{eq:A0}
	\Q(Y(T) = 0) = \Q(\tau\leq T) >0 \qquad \text{ and } \qquad \Q(\{Y(\tau-) > 0\} \cap \{\tau \leq T\}) = 0.
\end{align}
Since $Y$ was assumed to be $\Q$--martingale we also have $ \Q(\tau\leq T) <1$.

Since $Y$ is a $\Q$--martingale it generates a probability measure $\Prob$ (it corresponds to the $\Prob$ from the previous section) via the Radon-Nikodym derivative $\dd \Prob/\dd \Q = Y(T)$; the probability measure $\Prob$ is absolutely continuous with respect to $\Q$,  but not equivalent to $\Q$. 
\begin{lemma}\label{L1}  Under Assumption~\eqref{eq:A0}, the process $1/Y$ is a nonnegative $\Prob-$strict local martingale with $\Prob(1/Y(T)>0) = 1$.\end{lemma}
This statement of Lemma~\ref{L1}  follows directly from simple computations; see, for example, Theorem~2.1 in \citet{CFR2011}.

We introduce the following basic assumption:  
\begin{equation}
\begin{gathered} \label{eq:A1}
	\text{There exists $x \in (0,1)$ and an admissible strategy $H = (H(t))_{t \in [0,T]}$ s.t. } 
		 V^{x,H}(T) \geq \1_{\{Y(T)>0\}}.  
\end{gathered}
\end{equation}

Note that the market $(\Omega, \F, (\F(t))_{t \in [0,T]}, \Q), S$ of the last subsection satisfies NFLVR.  Thus, Assumption~\eqref{eq:A1} is equivalent to the assumption that the minimal superreplication price of the contingent claim $ \1_{\{Y(T)>0\}}$ is smaller than $1$ in this market; to wit,
\begin{align*}
	\sup_{R \in \mathcal{M}} \E^R[ \1_{\{Y(T)>0\}}]  < 1,
\end{align*}
where $\mathcal{M}$ denotes the set of all probability measures that are equivalent to $\Q$ and under which $S_i$ is a local martingale for each $i = 1, \ldots, d$.

We now state and prove the main result of this note:
\begin{theorem}\label{T1}
Under the setup of this section and under Assumptions~\eqref{eq:A0} and \eqref{eq:A1}, the market $(\Omega, \F, (\F(t))_{t \in [0,T]}, \Prob), S$ satisfies NUPBR but does not satisfy NFLVR. Moreover, any predictable process $H$ that satisfies the condition in Assumption~\ref{eq:A1} is a strong arbitrage strategy in the market $(\Omega, \F, (\F(t))_{t \in [0,T]}, \Prob), S$.
\end{theorem}
\begin{proof}
	Note that the process $H$ from Assumption~\eqref{eq:A1} is an admissible trading strategy under $\Q$, and thus, under $\Prob$, too.  Since $\Prob(\1_{\{Y(T)>0\}} = 1)=1$ and since $x<1$, the strategy $H$ is a strong arbitrage strategy under $\Prob$. Thus, the market $(\Omega, \F, (\F(t))_{t \in [0,T]}, \Prob), S$  does not satisfy NFLVR.
	
	To see that the market satisfies NUPBR, note that, by Lemma~\ref{L1},  the process $1/Y$ is a $\Prob$--local martingale and that $S_i/Y$ is also a $\Prob$--local martingale for each $i = 1, \ldots, d$  by a generalized Bayes' formula; see, for example, Proposition~2.3(iii) in \citet{CFR2011}. Thus, a local martingale deflator exists and, by Proposition \ref{P1},  the market satisfies NUPBR.
\end{proof}

Basically any market  $(\Omega, \F, (\F(t))_{t \in [0,T]}, \Prob), S$ that satisfies NUPBR but not NFLVR implies the existence of a probability measure $\Q$ and of a $\Q$--local martingale $Y$ that  satisfies \eqref{eq:A0} and such that $\dd \Prob/\dd \Q = Y(T)$;
see \citet{DS_Bessel}, \citet{Ruf_hedging}, and \citet{Imkeller_Perkowski}.
In this sense, Theorem~\ref{T1} provides the reverse direction and therefore may be considered a systematic construction of markets satisfying NUPBR but not NFLVR.

\section{Examples}\label{Examples}
In the setup of the previous section, we now discuss several examples for markets in which Assumptions~\eqref{eq:A0} and \eqref{eq:A1} hold. Theorem~\ref{T1} then proves that all those markets satisfy, under $\Prob$, NUPBR but not NFLVR. 

\begin{ex}
Let $Y$ denote a nonnegative $\Q$--martingale that satisfies $Y(0) = 1$ and Assumption~\eqref{eq:A0}. Let $S_1$ be the right-continuous modification of the process $(E^\Q[1_{\{Y(T)>0\}}|\F(t)])_{t \in [0,T]}$ and let $S_i$ for $i = 2, \ldots, d$ denote any $\Q$--local martingale. Then, clearly Assumption~\eqref{eq:A1} is satisfied and Theorem~\ref{T1} may be applied. \qed
\end{ex}

\begin{ex}  \label{ex:2}
 As a second example, worked out in detail in the dissertation of Chau Ngoc Huy, consider the case when $Y$ is a compensated $\Q$--Poisson process with intensity $\lambda \geq 1/T$ started in one, stopped when hitting zero or when it first jumps. Set $S_1 = Y$ and let $S_i$ for $i = 2, \ldots, d$ denote any $\Q$--local martingale. Clearly, Assumption~\eqref{eq:A0} (with $\Q(Y(T) = 0) = \exp(-1)$) and Assumption~\eqref{eq:A1} (with $x = 1-\exp(-1)$) hold, and thus, Theorem~\ref{T1} applies. \qed
\end{ex}

\begin{ex}
Let us now slightly generalize the previous Poisson setup of Example~\ref{ex:2}. Towards this end, we fix $\lambda \geq 1/T$ and let $F_{\min} \leq 1 \leq F_{\max}$ denote two strictly positive reals such that 
\begin{align*}
	x := \frac{F_{\max}}{F_{\min}} \left(1 - \exp\left(-\frac{1}{F_{\max}}\right)\right) < 1.
\end{align*}
Furthermore, we choose an arbitrary distribution function $F$ with support $[F_{\min}, F_{\max}]$ and expectation $1$.  Now, we let $Y$ denote a compensated compound Poisson process with intensity $\lambda$ and jump distribution $F$ (under the probability measure $\Q$), started in one and stopped when hitting zero or at its first jump.   As before, we set $S_1 = Y$ and let $S_i$ for $i = 2, \ldots, d$ denote any $\Q$--local martingale, without making any further assumptions on them.
Again, Assumption~\eqref{eq:A0} is clearly satisfied. 

As the nonnegative $\Q$--local martingale $Y$ might have jumps of different sizes, the process $S$ does not necessarily have the predictable martingale representation property.
However, we may check Assumption~\eqref{eq:A1} by hand. To make headway, we fix $\tau$ as the first hitting time of zero by the compensated compound Poisson process $Y$ and let $\rho$ denote its first jump time.  Note that $\tau \wedge \rho \leq 1/\lambda \leq T$ since at time $1/\lambda$ the compound Poisson process $Y$ has either made a jump or hit zero.  We define the $S$--integrable predictable process $H = (H_1, \ldots, H_d)$ by $H_i \equiv 0$ for all $i = 2, \ldots, d$, and, 
\begin{align*}
	H_1(t) = \frac{\exp\left(-\frac{1 - \lambda t}{F_{\max}}\right)}{F_{\min}} \1_{\{t  \leq \rho \wedge \tau\}}
\end{align*}
for all $t \in [0,T]$. Then, with $x$ as defined above, we obtain
\begin{align*}
	V^{x,H}(T) &= V^{x,H}(\tau \wedge \rho)  = x + \int_0^{\tau \wedge \rho}\frac{\exp\left(-\frac{1-\lambda t}{F_{\max}}\right)}{F_{\min}}\dd Y(t) \\
		&=  x - \lambda  \int_0^{\tau \wedge \rho}\frac{\exp\left(-\frac{1 - \lambda t}{F_{\max}}\right)}{F_{\min}} \dd t+  \frac{\exp\left(-\frac{ 1 - \lambda \rho}{F_{\max}}\right)}{F_{\min}} \Delta Y(\rho) \1_{\{\rho \leq \tau\}} \\
		&\geq \frac{F_{\max}}{F_{\min}} \left(1 - \exp\left(-\frac{1-\lambda(\tau \wedge \rho)}{F_{\max}}\right) \right)+  \exp\left(-\frac{1-\lambda {\rho}}{F_{\max}}\right) \1_{\{\rho \leq \tau\}} \\
		&= \left(\frac{F_{\max}}{F_{\min}} \left(1 - \exp\left(-\frac{1- \lambda \rho}{F_{\max}}\right) \right)+  \exp\left(-\frac{1- \lambda {\rho}}{F_{\max}}\right) \right)\1_{\{\rho \leq \tau\}} \\
		&\geq \1_{\{\rho \leq \tau\}} = \1_{\{Y(T)>0\}},
\end{align*}
where the first inequality follows from the observation that $\Delta Y(\rho)/ F_{\min} \1_{\{\rho \leq \tau\}} \geq  \1_{\{\rho \leq \tau\}}$, the equality just afterwards follows from the observation that $ \lambda \tau = 1$ on the event $\{\tau \leq \rho\}$, and the last inequality follows from the two facts that $F_{\max}/F_{\min} \geq 1$ and that the inequality $\lambda \rho \leq 1$ holds on the event $\{\rho \leq \tau\}$.
Thus, Assumption~\eqref{eq:A1} is satisfied and Theorem~\ref{T1} may be applied. \qed
\end{ex}

\begin{remark} \label{R:complete}
 Note that Assumption~\eqref{eq:A1} is always satisfied if the process $S$ has the predictable martingale representation property under $\Q$.   To wit,  Assumption~\eqref{eq:A1} is always satisfied if the market $(\Omega, \F, (\F(t))_{t \in [0,T]}, \Q), S$ is complete.
\qed
\end{remark}

\begin{ex} Consider the filtration $(\F(t))_{t \in [0,T]}$ to be generated by a $d$-dimensional $\Q$--Brownian motion $B = (B_i)_{i = 1, \ldots, d}$ with $B_i = (B_i(t))_{t \in [0,T]}$.  Let $\sigma = (\sigma(t))_{t \in [0,T]}$ denote a progressively measurable matrix-valued process of dimension  $d \times d$ such that $\sigma(t)$ is $\Q$--almost surely invertible for $Lebesgue$-almost every $t \in [0,T]$ and let $S(\cdot) = \int_0^\cdot \sigma(t) \dd B(t)$. Assume that the process $\sigma$ is chosen so that the price process $S$ is strictly positive.  Let $Y$ be a nonnegative local martingale hitting zero with positive probability, for example, the process $1 + B_1$ stopped when hitting zero.
Then,  the process $S$ has the predictable martingale representation property under $\Q$, and the construction of Section~\ref{S:construction} generates a market that satisfies NUPBR but not NFLVR; see Remark~\ref{R:complete}.  We also refer to Theorem~3 in \citet{DS_Bessel}, where a similar setup is discussed.\qed
\end{ex}

\bibliographystyle{apalike}
\setlength{\bibsep}{1pt}
\bibliography{aa_bib}

\end{document}